\documentclass{sig-alternate-ceb}

\pdfoutput=1

\conferenceinfo{Working Paper}{University of Toronto}
\toappear{Working Paper}

\usepackage{times,helvet,courier}
\usepackage{enumerate}
\usepackage{graphics,graphicx}
 
\newtheorem{theorem}{Theorem}

\newtheorem{corollary}[theorem]{Corollary}
\newtheorem{defn}[theorem]{Definition}
\newtheorem{proposition}[theorem]{Proposition}
\newtheorem{remark}[theorem]{Remark}
\newtheorem{observation}[theorem]{Observation}

\newcommand{\bfu}{\mathbf{u}}

\newcommand{\bfb}{\mathbf{b}}
\newcommand{\bfe}{\mathbf{e}}
\newcommand{\bfp}{\mathbf{p}}

\newcommand{\bfq}{\mathbf{q}}

\newcommand{\bfr}{\mathbf{r}}

\newcommand{\calX}{\mathcal{X}}
\newcommand{\calB}{\mathcal{B}}

\newcommand{\bbR}{\mathbb{R}}

\newcommand{\veps}{\varepsilon}

\newcommand{\srw}{\rightarrow}
\newcommand{\grad}{\nabla}

\newcommand{\tuple}[1]{\langle #1 \rangle}

\newcommand{\uppr}[1]{{#1\!\!\!\uparrow}}
\newcommand{\lowr}[1]{{#1\!\!\!\downarrow}}

\newcommand{\onenorm}[1]{||#1||_1}
\newcommand{\enorm}[1]{||#1||_2}
\newcommand{\enormsq}[1]{||#1||_2^2}

\newcommand{\commentout}[1]{}
\newcommand{\scite}{\cite}

\title{Eliciting Forecasts from Self-interested Experts: Scoring Rules
  for Decision Makers}

\numberofauthors{1}
\author{
  \alignauthor Craig Boutilier\\
    \affaddr{Department of Computer Science}\\
    \affaddr{University of Toronto, Canada}\\
    \email{cebly@cs.toronto.edu}
}

\begin{document}


\maketitle

\begin{abstract}
\small
Scoring rules for eliciting expert 
predictions of random variables are usually developed assuming
that experts derive utility \emph{only} from the quality of their predictions
(e.g., score awarded by the rule, or payoff in a prediction market).
We study a more realistic setting in which (a) the principal is a decision
maker and will take a decision based on the expert's prediction; and
(b) the expert has an inherent interest in the decision. For example,
in a corporate decision market, the expert may derive different
levels of utility from the actions taken by her manager. As a consequence
the expert will usually have an incentive to misreport her forecast
to influence the choice of the decision maker if typical
scoring rules are used. We develop a general model for
this setting and introduce the concept of a \emph{compensation rule}.
When combined with the expert's inherent utility for
decisions, a compensation rule
induces a \emph{net scoring rule} that behaves
like a normal scoring rule. Assuming full knowledge of expert
utility, we provide a complete characterization of
all (strictly) proper compensation rules. 
We then analyze the situation where the
expert's utility function is not fully known to the decision maker. 
We show bounds on: (a) expert incentive to misreport; (b) the degree
to which an expert will misreport; and (c) decision maker loss 
in utility due to such 
uncertainty. These bounds depend in natural ways on the degree of
uncertainty, the \emph{local} degree of convexity of net scoring 
function, and natural properties of the decision maker's utility function.
They also suggest optimization procedures for the design of
compensation rules.  Finally, we briefly discuss
the use of compensation rules as market scoring rules for
self-interested experts in a prediction market.
\end{abstract}

\section{Introduction}

Eliciting predictions of uncertain events from \emph{experts}
or other knowledgeable agents---or 
relevant information pertaining to events---is a fundamental 
problem of study in statistics, economics, operations research,
artificial intelligence and a variety of other areas
\cite{wolfers:predictionSurvey,chen-pennock:AIMag2011}. Increasingly, 
robust mechanisms for prediction are being developed, proposed
and/or applied in real-world domains ranging from elections and sporting
events, to events of public interest (e.g., disease spread
or terrorist action), to corporate decision making. Indeed, the very idea of
crowd-sourcing and information (or prediction) markets is predicated
on the existence of practical mechanisms for information elicitation
and aggregation.

A key element in any prediction mechanism involves providing an expert agent
with the appropriate incentives to reveal a forecast they believe
to be accurate. Many forms of ``outcome-based'' 
\emph{scoring rules}, either individual or market-based,
provide experts with incentives to (a) provide sincere 
forecasts; (b) invest effort to improve the accuracy of their
personal forecasts; and (c) participate in the mechanism if they believe 
they can improve the quality of the principal's forecast. However,
with just a few exceptions (see, e.g., \cite{Shi-Conitzer:wine09,othman-sandholm:aamas10,chen-kash:aamas11,Dimitrov-Sami:acmec08}, 
most work fails to account for
the ultimate use to which the forecast will be put.
Furthermore, even these models assume that the experts who provide
their forecasts derive no utility from the final forecast, or 
how it will be used, except insofar as they will be rewarded
by the prediction mechanism itself.

In many real-world uses of prediction mechanisms, this assumption 
is patently false.  Setting
aside purely informational and entertainment uses of information
markets, the principal is often interested in exploiting the elicited
forecast in order to make a \emph{decision} 
\cite{hanson-decisionMarkets99,othman-sandholm:aamas10,chen-kash:aamas11}. 
In corporate 
prediction markets, the principal may base strategic business decisions
on internal predictions of uncertain events. In a hiring committee,
the estimated probability of various candidates accepting offers (and
being given offers by competitors) will influence the order in which
(and whether) offers are made. Of course, other
examples abound. Providing appropriate incentives in
the form of scoring rules is often difficult in such settings,
especially when the outcome distribution is conditional on the
decision ultimately taken 
by the principal \cite{othman-sandholm:aamas10,chen-kash:aamas11}. 
However, just as critically,
in these settings, the experts whose forecasts are sought often
have \emph{their own interests} in seeing specific decisions being 
taken, interests that are not (fully) aligned with those of the principal.
For example, in a corporate setting, an expert from a certain
division may have an incentive to misreport demand for 
specific products,
thus influencing R{\&}D decisions that favor her division. In
a hiring committee, an committee member may misreport the odds 
that a candidate will accept a competing position in order to bias
the ``offer strategy'' in a way that favors his preferred candidate.

In this work, we develop what we believe to be
the first class of scoring rules that incentivizes truthful
forecasts even when experts have an interest in the decisions taken
by the principal, and hence would like to provide forecasts
that manipulate that decision ``in their favor.''
Other work has studied both decision making and
incentive issues in prediction markets, but
none that we are aware of addresses the natural question of expert
self-interest in the decisions of the principal. 

Hanson \cite{hanson-decisionMarkets99} introduced the term
\emph{decision markets} to refer to the broad notion of
prediction markets where experts offer forecasts for events conditional on 
some policy being adopted or a decision being taken.
Othman and Sandholm 
\scite{othman-sandholm:aamas10} provide the first explicit, formal
treatment of a principal who makes decisions
based on expert forecasts. They address
several key difficulties that arise
due to the conditional nature of forecasts, but
assume that the experts themselves 
have no direct interest in the decision that is taken.
Chen and Kash \scite{chen-kash:aamas11}
extend this model to a wider class of informational
settings and decision policies.
Dimitrov and Sami \scite{Dimitrov-Sami:acmec08}
consider the issue of strategic
behavior across multiple markets and the possibility that an
expert may misreport her beliefs in one market to manipulate prices
(and hence gain some advantage) in another. Similarly, Conitzer
\scite{Conitzer:uai09} explores strategic aspects of prediction markets
through their connections to mechanism design. While the mechanism
design approach could prove very useful for the problems we
address (see concluding remarks in Sec.~\ref{sec:conclude}),
Conitzer assumes an expert's utility is derived solely from
the payoff provided by the prediction mechanism.
Also related to the model we develop here is the analysis of Shi et 
al.~\cite{Shi-Conitzer:wine09}, who consider
experts that, once they report their forecasts, can take
action to alter
the probabilities of the outcomes in question. Unlike our model,
they do not consider expert utility apart
from the payoff offered by the mechanism (though, as in our model,
the principal does have a utility function that dictates the
value of an expert report). 

Our basic building block is a scoring rule for a
single expert who knows the principal's \emph{policy}---i.e., mapping
from forecasts to decisions---and where the principal knows the
expert's utility for decisions. We show that the scoring rule
must compensate the expert in a very simple and intuitive way
based on her utility function. Specifically, the principal
uses a \emph{compensation function} that, when added to the
\emph{inherent utility} the expert derives from the principal's
decision, induces a proper scoring rule.
In a finite decision
space, an expert's \emph{optimal} utility function
is piecewise-linear and convex in
probability space---we describe one natural scoring rule based
on this function that is proper, but not
strictly so. We then provide a complete characterization of all
proper compensation functions. We also characterize those which,
in addition, satisfy weak and strong \emph{participation constraints}
that ensure an expert will be sufficiently compensated to ``play the
game.''

We then provide a detailed analysis of both \emph{expert uncertainty} in 
the principal's policy, and \emph{principal uncertainty}
in the expert's utility for decisions.
First we observe that the expert need not know the principal's
policy prior to providing her forecast as long as she can verify the
decision taken after the fact. Second, we analyze the impact of
principal uncertainty regarding the expert's utility function.
In general, the principal cannot ensure truthful reporting.
However, we show that, given bounds on this
uncertainty, bounds can then be derived on all of the following: 
(i) the expert's incentive to
misreport; (ii) the deviation of the expert's misreported forecast
from its true beliefs; and (iii) the loss in utility the principal
will realize due to this uncertainty. The first two bounds rely on
the notion of \emph{strong convexity} of the net scoring function
induced by the compensation rule. The third bound uses natural
properties of the principal's utility function. Apart from
bounds derived from global strong convexity, we show
that these bounds can be significantly tightened using
\emph{local} strong convexity, specifically, by ensuring merely
that the net scoring function is sufficiently (and differentially)
strongly convex near the decision boundaries of the principal's
policy. These bounds suggest computational optimization methods for
for designing compensation rules (e.g., using
splines related to the principal's utility function).
We conclude by briefly discussing a market scoring rule (MSR) based 
on our one-shot compensation rule. Using this
MSR, the principal may need to provide
more generous compensation to each expert than in the one-shot case,
simply to ensure participation; but in some special cases,
no additional compensation is needed.

The paper is organized as follows.  We begin with a basic background
on scoring rules for prediction mechanisms in Sec.~\ref{sec:background}.
In Sec.~\ref{sec:scoringrules} we define our model for analyzing
the behavior of self-interested experts, 
introduce \emph{compensation rules}, and show
that the resulting \emph{net scoring function} can
be used to analyze expert behavior. We 
provide a complete characterization of (strict) proper compensation rules and 
and further characterize two subclasses of compensation
rules that satisfy the two participation constraints mentioned above.
In Sec.~\ref{sec:uncertain} we relax two assumptions in our model.
We first show the expert need to be aware of the principal's policy 
for our model to work. We
then consider a principal that has imperfect knowledge of
the expert's utility function, and using the notion of
(local and global) strong convexity derives bounds on the
expert's incentive to misreport and the impact on the quality
of the principal's decision. 
After a brief discussion of market scoring rules in Sec.~\ref{sec:msr},
we conclude in Sec.~\ref{sec:conclude} 
with a discussion of several avenues for future research.

\section{Background: Scoring Rules}
\label{sec:background}

We begin with a very brief review of relevant concepts from
the literature on scoring rules and prediction markets. For
comprehensive overviews, see the surveys 
\cite{wolfers:predictionSurvey,chen-pennock:AIMag2011}.

We assume that an agent---the \emph{principal}---is interested
in assessing the distribution of some discrete random variable $\calX$ with 
finite domain $X = \{x_1, \ldots, x_m\}$. Let $\Delta(X)$ denote the set of
distributions over $X$, where $\bfp\in\Delta(X)$ is a nonnegative
vector $\tuple{p_1, \ldots, p_m}$ s.t.\ $\sum_i p_i = 1$.
The principle can engage one or more experts to provide a forecast 
$\bfp \in \Delta(X)$.  We focus first on the case of a single expert 
$E$.  For instance, to consider a simple toy example we
use throughout the sequel, the mayor of a 
small town may ask the local weather forecaster to 
offer a probabilistic estimate of weather
conditions for the following weekend.\footnote{More significant
examples in the domains of public policy or corporate decision
making, as discussed above, can easily be constructed by the
reader.}

We assume $E$ has beliefs $\bfp$ about $\calX$, but a key question is
how to incentivize $E$ to report $\bfp$ faithfully (and devote
reasonable effort to developing \emph{accurate} beliefs).
A variety of \emph{scoring rules} have been
developed for just this purpose 
\cite{brier:1950,savage:jasa71,mccarthy:pnas56,gneiting:jasa07}. A
\emph{scoring rule} is a function $S: \Delta(X) \times X\srw\bbR$
that provides a score (or payoff) $S(\bfr,x_i)$ to $E$ if she
reports forecast $\bfr$ and the realized outcome of $\calX$ is $x_i$,
essentially rewarding $E$ for her predictive ``performance'' 
\cite{brier:1950,mccarthy:pnas56}.
If $E$ has beliefs $\bfp$ and reports $\bfr$, her expected score
is $S(\bfr,\bfp) = \sum_i S(\bfr,x_i) p_i$. We say $S$ is a
\emph{proper scoring rule} iff a truthful report is optimal for $E$:
\begin{align}
S(\bfp,\bfp) \geq S(\bfr,\bfp), \quad \forall \bfp,\bfr\in \Delta(X)
\label{eq:properscore}
\end{align}
We say that $S$ is \emph{strictly proper} if inequality~(\ref{eq:properscore})
is strict for $\bfr \neq \bfp$ (i.e., $E$ has strict disincentive to
misreport). A variety of strictly proper scoring rules have been
developed, among the more popular being the log scoring rule,
where $S(\bfp,x_i) = a \log p_i + b_i$ (for arbitrary constants
$a > 0$ and $b_i$) \cite{mccarthy:pnas56,savage:jasa71}. 
In what follows, we will restrict attention to \emph{regular}
scoring rules in which payment $S(\bfr,x_i)$ is bounded whenever
$r_i > 0$.

Proper scoring rules can be fully characterized in terms of
convex cost functions \cite{mccarthy:pnas56,savage:jasa71}; here we review the 
formulation of Gneiting and Raftery \scite{gneiting:jasa07}.
Let $G: \Delta(X) \srw \bbR$ be any convex function over 
distributions---we refer to $G$ as a \emph{cost function}. We denote
by $G^*: \Delta(X) \srw \bbR^m$ some \emph{subgradient} of $G$,
i.e., a function satisfying 
$$G(\bfq) \geq G(\bfp) + G^*(\bfp)\cdot(\bfq-\bfp)$$
for all $\bfp,\bfq\in\Delta(X)$.\footnote{If $G$ is differentiable at
$\bfp$ then the subgradient at that point is unique, namely,
the gradient $\grad G(\bfp)$.} Such cost functions and associated
subgradients can be used to derive any proper scoring rule.
\begin{theorem} \cite{mccarthy:pnas56,savage:jasa71,gneiting:jasa07}
\label{thm:savage}
A regular scoring rule $S$ is proper iff 
\begin{align}
S(\bfp,x_i) = G(\bfp) - G^*(\bfp)\cdot(\bfp) + G^*_i(\bfp)
  \label{eq:subgradient}
\end{align}
for some convex $G$ and subgradient $G^*$. $S$ is strictly proper
iff $G$ is strictly convex. 
\end{theorem}
Intuitively, Eq.~\ref{eq:subgradient}
defines a hyperplane
$$H_{\bfp} = \tuple{S(\bfp,x_1),\ldots,S(\bfp,x_m)},$$
for each point $\bfp$, that is subtangent to $G$ at $\bfp$. This
defines a linear function, for any fixed report $\bfp$, giving
the expected score of that report given beliefs 
$\bfq$: $S(\bfp,\bfq) = H_{\bfp}\cdot\bfq$.
An illustration is given in Fig.~\ref{fig:subgradient} for
a simple one-dimensional (two-outcome) scenario.

\begin{figure}
\centering
\includegraphics[scale=0.46]{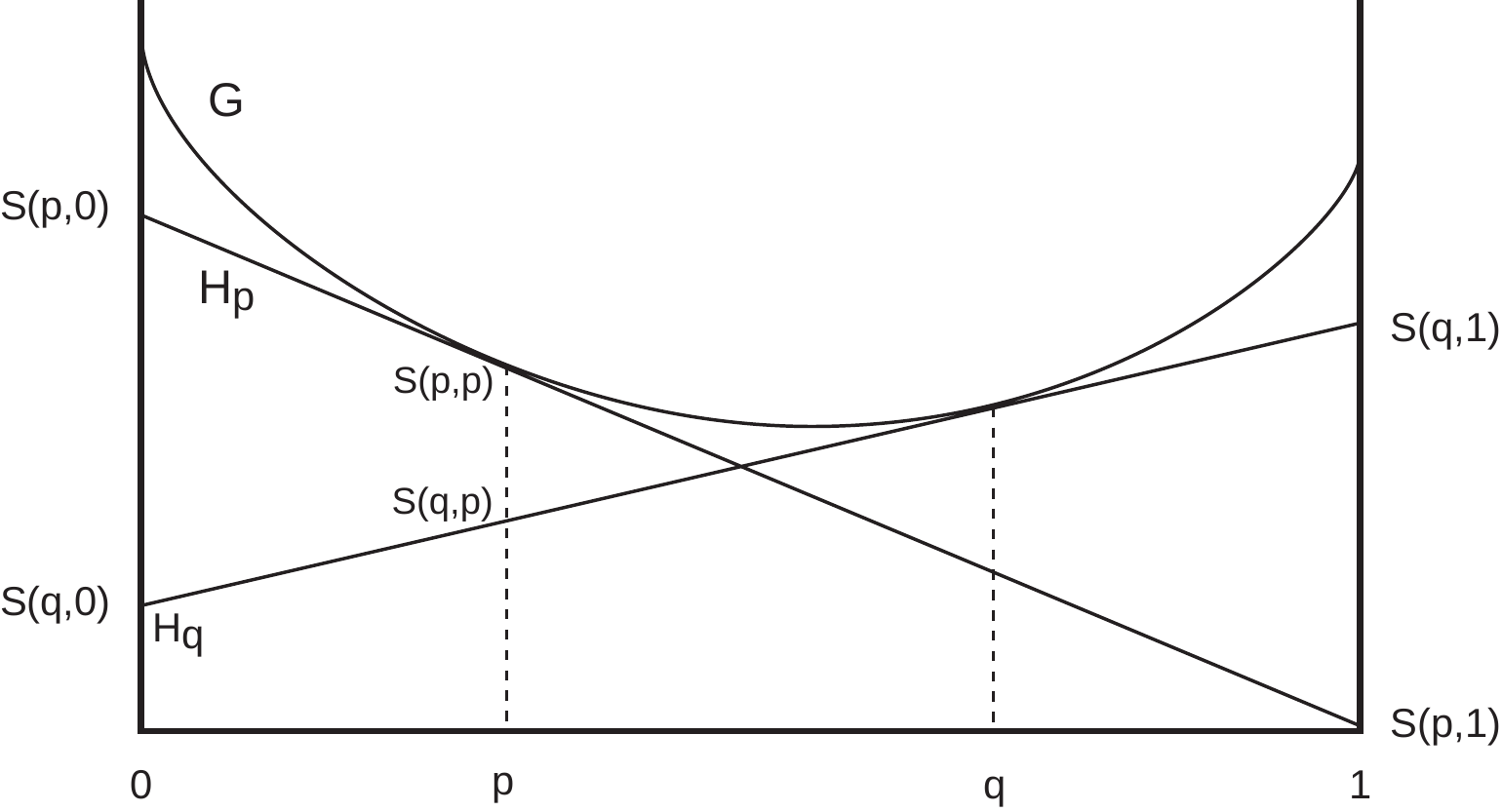}
\vskip -2mm
{\caption{\textrm\small Illustration of a proper scoring rule. If 
  the expert reports 
  $p$, its expected score (relative to its true beliefs) is given by 
  the subtangent hyperplane $H_p$. For any report $q$ different from
  $p$, the expected score $S(q,p) = H_q\cdot p$ must be less than the 
  expected score $S(p,p) = H_p\cdot p$ of truthful reporting.}
\label{fig:subgradient}
}
\vskip -2mm
\end{figure}

There are a number of prediction market mechanisms that allow the
principal to extract information from multiple experts; see
\cite{wolfers:predictionSurvey,chen-pennock:AIMag2011} for
excellent surveys. Here we focus on \emph{market scoring rules
(MSRs)}
\cite{hanson-scoringrules:isf03,hanson-scoringrules:jpm07}, which allow
experts to (sequentially) change the forecasted $\bfp$ using any
proper scoring rule $S$. Given the current forecast $\bfp'$, an
expert can change the forecast to $\bfp$ if she is willing
to pay according to $S(\bfp',\cdot)$ and receive payment $S(\bfp,\cdot)$.
If her true beliefs $\bfp$ are different from $\bfp'$
and the scoring rule is strictly proper, then
she has incentive to participate and report truthfully.
Under certain conditions, MSRs can be interpreted as automated
market makers \cite{chen-pennock:uai07}. Since each expert pays
the amount due to the previous expert for her prediction, the net
payment of the principal is the score associated with the 
final prediction.

\section{Scoring Rules for Self-interested Experts}
\label{sec:scoringrules}

Scoring rules in standard models assume that an expert offering 
a forecast
is uninterested in any aspect her forecast other than the score
she will derive from her prediction.
As discussed above, there are many settings where the principal
will make a decision based on the received forecast, and the expert has
a direct interest in this decision. In this section, we develop
a model for this situation and devise a class of scoring rules that
incentive self-interested agents to report their true beliefs.

\subsection{Model Formulation}
\label{sec:model}

We assume the principal, or \emph{decision maker (DM)}, will elicit
a forecast of $\calX$ from expert $E$, and make a decision that
is influenced by this forecast. Let $D = \{d_1, \ldots, d_n\}$
be the set of possible decisions, and $u_{ij}$ be DM's utility should
it take decision $d_i$ with $x_j$ being the realization of $\calX$.
Letting $\bfu_i = \tuple{u_{i1},\cdots,u_{im}}$, the expected utility
of decision $d_i$ given distribution $\bfp$ is
$U_i(\bfp) = \bfu_i \cdot \bfp$. For any beliefs $\bfp$, DM will
want to take the decision that maximizes expected utility, giving
DM the \emph{utility function}
$U(\bfp) = \max_{i\leq n} U_i(\bfp)$.  Since each $U_i$ is a linear
function of $\bfp$, $U$ is piecewise linear and convex (PWLC).
Furthermore, each $d_i$ is optimal in a (possibly empty) convex region
of belief space 
$D_i = \{\bfp : \bfu_i \cdot \bfp \geq \bfu_j \cdot \bfp, \forall j \}.$
We assume DM acts optimally and that it has a policy
$\pi: \Delta(X) \srw D$ that selects some optimal decision 
$\pi(\bfp)$ for any expert forecast $\bfp$. In what follows, we
take $D_i = \pi^{-1}(d_i)$. We denote by $D_{ij}$ the
(possibly empty) \emph{boundary} between $D_i$ and $D_j$. Notice
that for any $\bfp\in D_{ij}$ we must have 
$U_i(\bfp) = U_j(\bfp)$.\footnote{Assuming ``ties'' at boundaries
are broken consistently, the
regions $D_i$ will be convex, but possibly open.}

In our running example, suppose the mayor must decide whether to
hold a civic ceremony in the town park or at 
a private banquet facility.
Given a forecast probability $p$ of rain, she will make
outdoor arrangements at the park if $p$ falls below some
threshold $\tau$, and will rent the banquet facility if $p$ is above
$\tau$ (here $\tau$ is the indifference probability:
$U_{\mbox{\scriptsize park}}(\tau) = U_{\mbox{\scriptsize banq}}(\tau)$).

We note that this model of DM utility is slightly more
restricted than that of \cite{othman-sandholm:aamas10,chen-kash:aamas11}, who
allow the utility of each decision to depend on a different
random variable, and assume that the realization of a variable 
will be observed \emph{only if} the corresponding decision is taken.
This introduces difficulties in offering suitable incentives
for participation that do not arise in our setting; indeed, the
primary contribution of Othman and Sandholm \scite{Conitzer:uai09},
and the extension by Chen and Kash \scite{chen-kash:aamas11}, is a
characterization of a form of proper scoring in the face of these
complications.  We also
confine our attention primarily to a principal that maximizes its
expected utility given $E$'s report (in the terminology of
\cite{othman-sandholm:aamas10}, DM uses the \emph{max decision rule}),
though we remark on the possible use of stochastic policies 
by DM in Sec.~\ref{sec:utilityuncertain}.

Now suppose a single expert $E$ is asked to provide a forecast of
$\calX$ that permits DM to make a decision. Assume that $E$ knows
DM's policy $\pi$: knowledge of DM's utility function is sufficient
for knowledge of the policy but is not required (we discuss
the possibility of $E$ being uncertain about $\pi$ in
Sec.~\ref{sec:policyuncertain}). Further, assume
that $E$ has its own utility function or \emph{bias} $b$,
where  $b_{ij}$ is $E$'s utility should DM take decision $d_i$ and $x_j$
is the realization of $\calX$. Define 
$\bfb_i = \tuple{b_{i1},\cdots,b_{im}}$; and 
let $E$'s expected utility for $d_i$ given $\bfp$ be
$B_i(\bfp) = \bfb_i \cdot \bfp$.  
In our small example, the weather forecaster may be
related to the owner of the banquet facility, and will get some 
degree of satisfaction (or a small kickback) if the mayor's
ceremony is held there.

As with DM, $E$'s \emph{optimal utility
function} $B^*$ (if DM were acting on $E$'s behalf)
is PWLC:
\begin{align}
B^*(\bfp) = \max_i \bfb_i \cdot \bfp.
\label{eq:expOptUtil}
\end{align}
Denote by $D^*(\bfp)$ the decision $d_i$ that maximizes
Eq.~\ref{eq:expOptUtil}, i.e., $E$'s preferred decision given
beliefs $\bfp$ (see Fig.~\ref{fig:expertUtility} for an
illustration).

Of course, DM is pursuing its own policy $\pi$, not acting to
optimize $E$'s utility.  Hence $E$'s actual
utility for a specific report $\bfr$ under beliefs $\bfp$ is
given by 
\begin{align}
B^\pi(\bfr,\bfp) = \bfb_{\pi(\bfr)}\cdot\bfp;
\label{eq:expUtil}
\end{align}
that is, if she reports $\bfr$,
DM will take decision $\pi(\bfr) = d_k$ for some $k$,
and she will derive benefit 
$B_k(\bfp)$. We refer to $B^\pi(\bfr,\bfp)$ as $E$'s 
\emph{inherent utility} for reporting $\bfr$. Similarly,
$B(\bfr,x_i) = b_{i,\pi(\bfr)}$ is $E$'s inherent utility for report $\bfr$ 
under realization $x_i$. This is simply the inherent benefit she derives from 
the decision she induces DM to take. 
This is illustrated in Fig.~\ref{fig:expertUtility}.
Notice that $E$'s utility for reports, given any fixed beliefs
$\bfp$, is not generally continuous, with potential (jump) discontinuities
at DM's decision boundaries.

\begin{figure}
\centering
\includegraphics[scale=0.41]{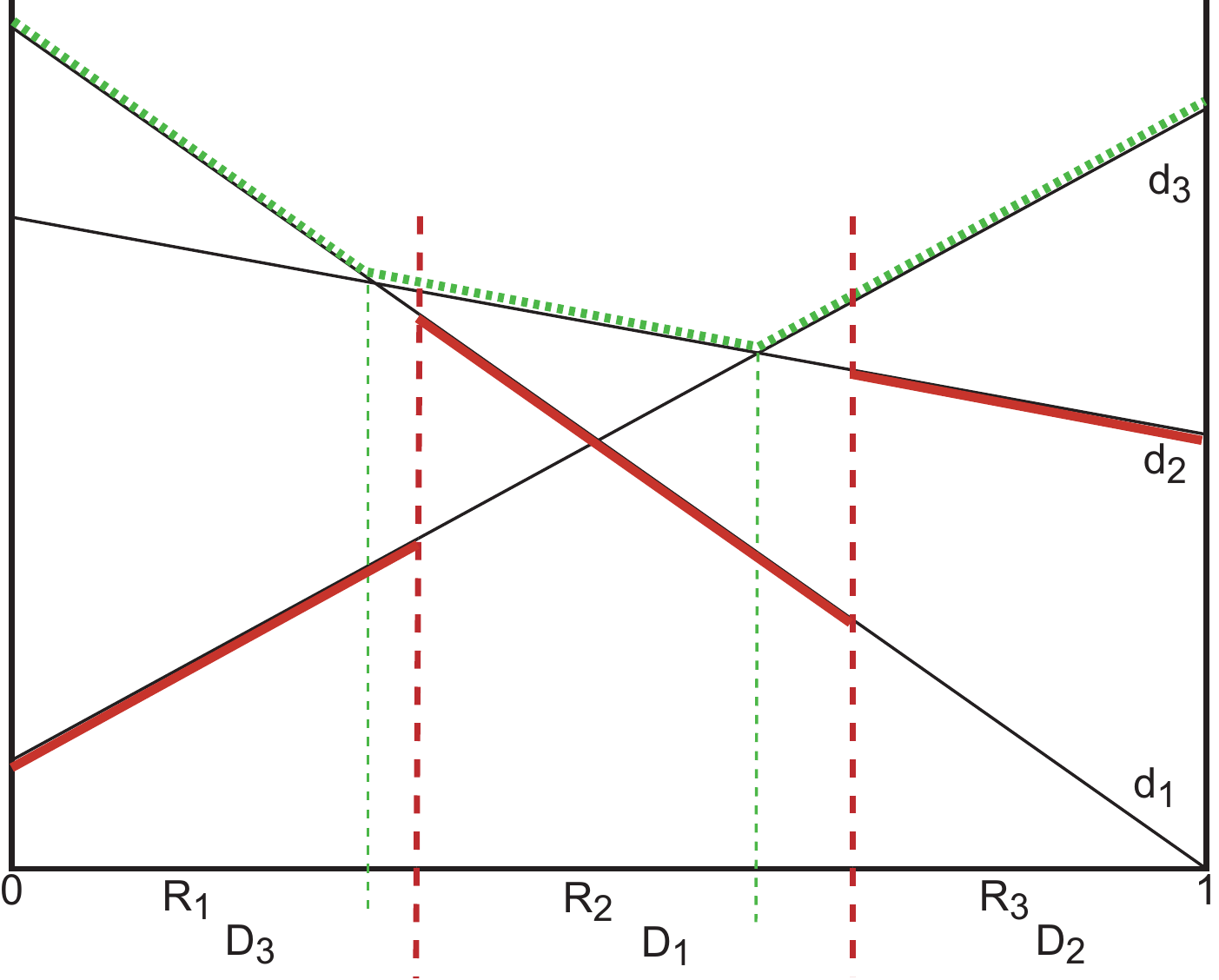}
\vskip -2mm
\textrm\small{\caption{Expert utility function. $E$'s utility for
  each decision $d_i$ is given by the corresponding hyperplane
 (here, thin black line).  $E$'s ``optimal'' utility $B^*$
  is the PWLC function shown by the dotted green line
  (i.e., the upper surface), with
  each $R_i$ denoting the regions of belief 
  space where $D^*(\bfp) = d_i$.  Regions $D_i$ represent
  DM's policy, where $\pi(\bfp) = d_i$ for $\bfp\in D_i$.
  The thick red lines denote the (discontinuous) utility 
  $B^\pi$ that $E$
  will receive from (truthfully) reporting her belief.}
\label{fig:expertUtility}
}
\vskip -2mm
\end{figure}

Without some scoring rule, there is a clear incentive for $E$ to
misreport its true beliefs to induce DM to take a decision that $E$
prefers, thereby causing DM to take a suboptimal decision.
For instance, in
Fig.~\ref{fig:expertUtility}, if $E$'s
 true beliefs $\bfp$ lie in $R_1$,
its preferred decision is $d_1$; but truthful reporting will induce
DM to take decision $d_3$. $E$ has greater inherent utility for
reporting (any) $\bfr \in D_1$. Indeed, its gain from the misreport
is $\bfp \cdot (\bfb_1 - \bfb_3)$. Equivalently, 
$E$ stands to lose $\bfp \cdot (\bfb_1 - \bfb_3)$ by reporting truthfully.
Intuitively, a proper scoring rule would remove this incentive
to misreport.

\subsection{Compensation Rules}

If DM knows $E$ utility function, it could reason about
$E$'s incentive to misreport and revise its decision policy
accordingly. Of course, this would naturally lead to a Bayesian
game requiring analysis of its Bayes-Nash equilibria, and
generally leaving DM with uncertainty about $E$'s true 
beliefs.\footnote{See Dimitrov and Sami \scite{Dimitrov-Sami:acmec08}
and Conitzer \scite{Conitzer:uai09} for just such a game-theoretic treatment
of prediction markets (without decisions).}
Instead, we wish to derive a scoring rule that DM can use to incentivize
$E$ to report truthfully. 

A \emph{compensation function} $C:\Delta(X)\times X\srw \bbR$ is 
a mapping from reports and outcomes into payoffs, exactly like a
standard scoring rule. Unlike a scoring rule, however, $C$ does not
fully
determine $E$'s utility for a report; one must also take into account
the inherent utility $E$ derives from the decision it prompts the
DM to take.
Any compensation function $C$ induces a \emph{net scoring function}:
\begin{align}
S(\bfp,x_i) &= C(\bfp,x_i) + B^\pi(\bfp,x_i)
\label{eq:netscore}
\end{align}
$E$'s expected net score for report $\bfr$ under beliefs $\bfp$
is $S(\bfr,\bfp) = C(\bfr,\bfp) + B^\pi(\bfr,\bfp)$, where
$C(\bfr,\bfp) = \sum_i p_i C(\bfr,x_i)$ is $E$'s expected compensation.

One natural way to structure the compensation
function is to use $C$ to compensate $E$ 
for the loss in inherent utility incurred by reporting 
its true beliefs $\bfp$ (relative to its best report).
This would remove any incentive for $E$ to misreport.
We define a particular compensation function $C_1$ that accounts
for this loss:
\begin{align}
C_1(\bfp,x_i) = b_{i,D^*(\bfp)} - b_{i,\pi(\bfp)}.
\label{eq:c1}
\end{align}
$C_1(\bfp,x_i)$ is simply the difference between $E$'s 
\emph{realized} utility for its optimal decision (relative to
its report $\bfp$) and the actual decision she induced.
$C_1$ does not satisfy the usual properties of scoring rules: it
is given by a subgradient of the loss function, which is not convex,
nor even continuous. However, $E$'s payoff for a report consists of
\emph{both} this compensation and its inherent utility, i.e., her net
score:
\begin{align}
S_1(\bfp,x_i) &= C_1(\bfp,x_i) + B^\pi(\bfp,x_i)\\
            &= (b_{i,D^*(\bfp)} - b_{i,\pi(\bfp)}) + b_{i,\pi(\bfp)})\\
            &= b_{i,D^*(\bfp)}
\label{eq:netscoreone}
\end{align}
$E$'s expected net score under beliefs $\bfp$ is identical to her
expected utility for the optimal decision $D^*(\bfp)$.
Hence, no other report can induce a decision
that gives her greater utility. Informally, this shows that truthful
reporting is optimal. It can be seen directly by observing that
the net score $S_1$ can be derived from Eq.~\ref{eq:subgradient}
by letting $G(\bfp) = 
  B^*(\bfp) = \max_{i\leq n} B_i(\bfp)$ be $E$'s optimal
utility function (which is PWLC, hence convex), and using the
subgradient $G^*(\bfp)$ given by the hyperplane corresponding to the
optimal decision $D^*(\bfp)$ at that point.\footnote{At interior
points of $E$'s decision regions, the hyperplane is the unique
subgradient. At $E$'s decision boundaries,
an arbitrary subgradient can be used.}

\begin{defn}
\label{defn:properC} 
A compensation function $C$ is \emph{proper} iff the expected
net score function $S$ satisfies $S(\bfp,\bfp) \geq S(\bfq,\bfp)$
for all $\bfp,\bfq\in\Delta(X)$. $C$ is \emph{strictly proper} if
the inequality is strict.
\end{defn}
We don't prove this formally since we prove a more general result
below, but the above informal argument shows:
\begin{proposition}
Compensation function $C_1$ is proper.
\end{proposition}

\begin{remark} We've defined the compensation function using the 
space of all decisions $D$. However, this may cause DM to 
compensate $E$ for
decisions it will never take. If we restrict attention
to those decisions in the range of DM's policy $\pi$,  then
the above characterization still applies (and will typically
reduce total compensation). In what follows, we
assume the set of decisions has been pruned to include
only those $d\in D$ for which $\pi^{-1}(d) \neq\emptyset$, and
that $E$'s utility function is defined relative to that set.
\end{remark}

Compensation function $C_1$, while proper, is
not strictly proper.  The induced net scoring function
$S_1$ is characterized by a non-strictly convex 
cost function $G$, since $G=B^*$. In particular, 
for any region $R(d)$ of
belief space where a single decision $d$ is optimal for $E$,
every report $\bfp\in R(d)$ has the same expected net score,
hence there is no ``positive'' incentive for truthtelling. 

While $C_1$ gives us one mechanism for proper scoring with
self-interested experts, we can generalize the approach to
provide a complete characterization of all proper (and strictly
proper) compensation functions. We derived $C_1$ by
compensating $E$ for its \emph{loss} due to truthful
reporting. This approach is more ``generous'' than necessary.
Rather than compensating $E$ for
its loss, we need only \emph{remove the potential gain} from misreporting.
The key component of $C_1$ is not the 
``compensation term'' $b_{i,D^*(\bfp)}$, but rather the
the penalty term $-b_{i,\pi(\bfp)}$. It is this penalty
that prevents $E$ from benefiting by changing DM's decision. Any
such gain is subtracted from
its compensation by the inclusion of 
$-b_{i,\pi(\bfp)}$. We insist only that the positive compensation
term is convex: it need bear no connection to $E$'s actual utility
function to incentivize truthfulness.\footnote{Incentive to
\emph{participate} is discussed below.}
Indeed, we can fully characterize the space of proper and
strictly proper compensation functions:
%
%
\begin{theorem}
\label{thm:characterize}
A compensation rule $C$
is proper for $E$ iff
\begin{align}
C(\bfp,x_i) = G(\bfp) - G^*(\bfp)\cdot\bfp + G^*_i(\bfp) - b_{i,\pi(\bfp)}
\label{eq:convexcompensation}
\end{align}
for some convex function $G$, and subgradient $G^*$ of $G$.
$C$ is strictly proper iff $G$ is strictly convex.
\end{theorem}
\begin{proof}
Suppose $C$ is given by Eq.~\ref{eq:convexcompensation}. $E$'s
utility for a report $\bfp$ given outcome $x_i$ is given by
its net score:
\begin{align*}
S(\bfp,x_i) &= C(\bfp,x_i) + B^\pi(\bfp,x_i)\\
            &= G(\bfp) - G^*(\bfp)\cdot\bfp + G^*_i(\bfp) - b_{i,\pi(\bfp)}
               + b_{i,\pi(\bfp)}\\
            &= G(\bfp) - G^*(\bfp)\cdot\bfp + G^*_i(\bfp)
\end{align*}
Since $S$ satisfies the conditions of Thm.~\ref{thm:savage}, the
standard proof of propriety of $S$ can be used. Similarly, if $G$ is
strictly convex, $S$ is strictly proper.

Conversely, suppose $C$ is proper (so that the induced net score
satisfies $S(\bfp,\bfp) \geq S(\bfq,\bfp)$). Define $G(\bfp) = S(\bfp,\bfp)$.
If $S$ is proper in this
sense, it is easy to show that $S(\bfq,\bfp)$ is a convex
function of $\bfq$ (for fixed beliefs $\bfp$); and since 
$G(\bfp) = S(\bfp,\bfp) = \max_{\bfq} S(\bfq,\bfp)$ is
the maximum of a set of convex functions (where
the last equality holds because $C$ is proper), $G$ is itself
convex. Thm.~\ref{thm:savage} (or more precisely the method
used to prove it) ensures that, for some
subgradient $G^*$ of $G$, we have
$S(\bfp,x_i) = G(\bfp) - G^*(\bfp)\cdot\bfp + G^*_i(\bfp)$.
Hence,
\begin{align*}
C(\bfp,x_i) &= S(\bfp,x_i) - B^\pi(\bfp,x_i)\\
            &= G(\bfp) - G^*(\bfp)\cdot\bfp + G^*_i(\bfp) - b_{i,\pi(\bfp)}
\end{align*}
so $C$ has the required form.
If $C$ is strictly proper, then $G$ must be
strictly convex and Thm.~\ref{thm:savage} can again be applied.
\end{proof}
An illustration of a cost function $G(\bfp)$
that gives rise to a proper compensation function
is shown in Fig.~\ref{fig:compensation}(a).

\begin{figure*}
\centering
\includegraphics[scale=0.83]{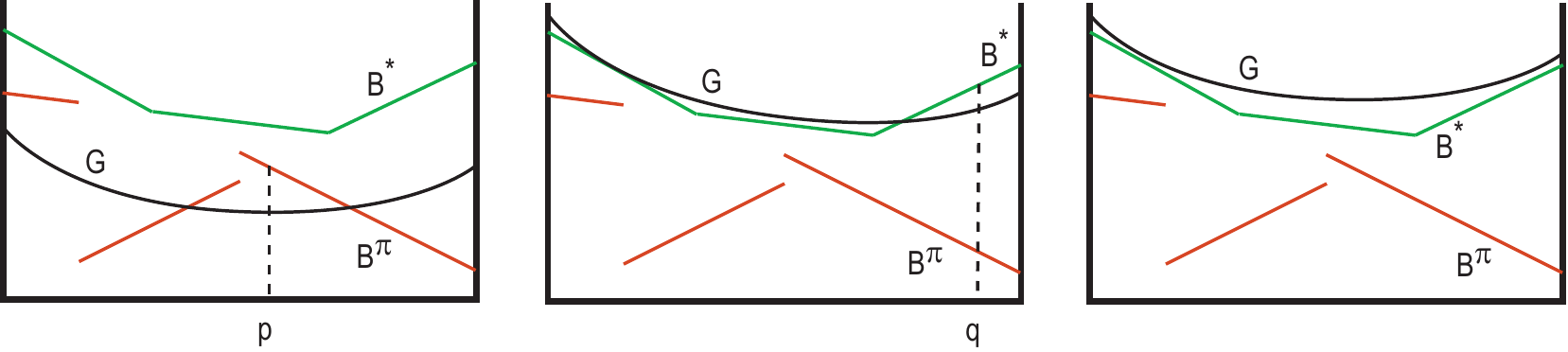}
\vskip -5mm
\textrm\small{\caption{Illustration of the cost functions $G$ that
correspond to strictly proper compensation rules. In each figure,
$E$'s optimal utility $B^*$ is the PWLC function
shown in green, and $E$'s inherent
utility $B^\pi$ is the discontinuous function in red. 
The net scoring function
$G(\bfp) = S(\bfp,\bfp)$, the convex curve, induces an
expected compensation function $C(\bfp,\bfp)$ by subtracting
$B^\pi$. (a) A strictly proper rule that violates
weak participation at point $p$. (b) A rule that satisfies
weak participation but violates strong participation at
point $q$. (c) A rule that satisfies strong participation.}
\label{fig:compensation}
}
\vskip -2mm
\end{figure*}

The characterization of Thm.~\ref{thm:characterize} ensures truthful
reporting, but may not provide incentives for participation. Indeed,
the expert may be forced to pay the DM \emph{in expectation} for
certain beliefs. Specifically, if $G(\bfp) < B^\pi(\bfp)$,
$E$'s expected compensation $C(\bfp,\bfp)$ is negative. Unless
the DM can ``force'' $E$ to participate, this will cause $E$ to
avoid providing a forecast if its beliefs are $\bfp$ (e.g., see point $p$
is Fig.~\ref{fig:compensation}(a)). In general,
we'd like to provide $E$ with non-negative expected compensation. We
can do this by insisting that the compensation rule weakly
incentives participation:
\begin{defn}
\label{defn:weakC} 
A compensation function $C$ satisfies
\emph{weak participation} iff for any beliefs $\bfp$,
$E$'s expected compensation for truthful reporting
$C(\bfp,\bfp)$ is non-negative.
\end{defn}
\noindent
(See Fig.~\ref{fig:compensation}(b) for an illustration
of a cost function $G$ that induces a compensation rule $C$
satisfying weak participation.)
\begin{theorem}
\label{thm:characterizeWeak}
A proper compensation rule $C$ satisfies
weak participation iff it meets the conditions of
Thm.~\ref{thm:characterize} and $G(\bfp) \geq B^\pi(\bfp)$
for all $\bfp\in\Delta(X)$.
\end{theorem}
\begin{proof}
The proof is straightforward: if $G(\bfp) \geq B^\pi(\bfp)$
for all $\bfp$, then for any truthful report
$\bfp$ $E$'s expected compensation is 
$G(\bfp) - B^\pi(\bfp) \geq 0$. Conversely, if 
$G(\bfp) < B^\pi(\bfp)$ for some $\bfp$, then if $E$ holds
beliefs $\bfp$, a truthful report has negative expected
compensation $G(\bfp) - B^\pi(\bfp) < 0$.
\end{proof}

While weak participation seems desirable, even it is not
strong enough to ensure an expert's participation in the
mechanism in general. Suppose we define a compensation function using
some convex cost function $G(\bfp)$. If 
$E$ participates, she will maximize
her net payoff by reporting her true beliefs, say, $\bfp$.
But suppose that $G(\bfp) < B^*(\bfp)$.
While $E$ may not be certain how DM will act without its input 
(e.g., she may not know DM's ``default beliefs'' precisely), she
may nevertheless have beliefs about DM's default policy.
And, if $E$ believes DM will take decision $D^*(\bfp)$ if she
provides no forecast, then she will be better off not
participating and taking the expected payoff $B^*(\bfp)$
derived solely from her inherent utility, and forego participation
in the mechanism
(which limits her expected payoff to $G(\bfp)$). (See point $q$ in
Fig.~\ref{fig:compensation}(b).) To prevent
this we can require that $C$ \emph{strongly} incentivize
participation, by insisting no matter what $E$ believes about
DM's default policy (i.e., its action given no reporting), it 
will not sacrifice expected utility by participating
in the mechanism.
\begin{defn}
\label{defn:strongC} 
A compensation function $C$ satisfies
\emph{strong participation} iff, for any decision $d_i\in D$,
for any beliefs $\bfp$,
$E$'s net score for truthful reporting is no less
than $B_i(\bfp)$.
\end{defn}
Strong participation means that $E$ has no incentive to
abstain from participation (and need not ``take its chances'' that
DM will make a decision it likes).
This definition is equivalent to requiring that $E$'s
expected utility for truthful reporting, as a 
function of $\bfp$ is at least as great as her optimal utility function,
i.e., $S(\bfp,\bfp) \geq B^*(\bfp)$ for all $\bfp\in\Delta(X)$.
Fig.~\ref{fig:compensation}(c) illustrates such a compensation
rule. 
\begin{theorem}
\label{thm:characterizeStrong}
Proper compensation rule $C$ satisfies
strong participation iff it meets the conditions of
Thm.~\ref{thm:characterize} and 
$G(\bfp) \geq B^*(\bfp)$ for all $\bfp\in\Delta(X)$.
\end{theorem}
\begin{proof}
The proof is straightforward. Suppose $G(\bfp) \geq B^*(\bfp)$
for all $\bfp$. If $E$ holds beliefs
$\bfp$, then a truthful report has expected net score of
$G(\bfp) \geq B^*(\bfp)$, and for no beliefs about DM's default 
policy can $E$ derive higher utility by not participating.
Conversely, suppose
$G(\bfp) < B^*(\bfp)$ for some $\bfp$. If $E$ holds
beliefs $\bfp$ and also believes that DM will take action
$d = D^*(\bfp)$ if $E$ does not report, then $E$ will derive
utility $B^*(\bfp)$ by not participating, better than the
optimal expected score $G(\bfp)$ from participating.
\end{proof}

\begin{observation}
\label{obs:minimalstrong}
Compensation rule $C_1$ is the unique minimal
(non-strictly) proper rule satisfying
strong participation. That is, no compensation rule offers
lower compensation for any report without violating strong
participation.
\end{observation}

In general, if we insist on strong participation,
DM must provide potential compensation up to the level of
$E$'s maximum utility \emph{gap}:
$$g(B) = \max_{i\leq m, j,k\leq n} b_{ik} - b_{ij}.$$
However, this degree of compensation is needed only
if DM and $E$ have ``directly conflicting'' interests (i.e.,
DM takes a decision whose realized utility is as far from optimal
as possible from $E$'s perspective). In such cases, one would expect 
$E$'s utility to be significantly less than DM's. If not,
this compensation would not be worthwhile for DM.
Conversely, if $E$'s interests are well aligned with those of DM,
the total compensation required will be small. The most extreme case
of well-aligned utility is one where functions $\pi$ and $D^*$ coincide, i.e., 
$\pi(\bfp) = D^*(\bfp)$ for all beliefs $\bfp$, in which case,
no compensation is required. Specifically,
compensation function $C_1(\bfp) = 0$ for all $\bfp$; and while
$C_1$ is not strictly proper, the only misreports that $E$ will
contemplate (i.e., that do not reduce its net score)
are those that cannot change DM's decision (i.e., cannot impact DM's
utility). As a consequence, DM should elicit forecasts from an
expert who either (a) has well-aligned interests in the decisions
being contemplated; (b) has interest whose magnitude is small
(hence requires modest compensation) relative to DM's own utility;
or (c) can be ``forced'' to make a prediction (possibly
at negative \emph{net} cost).\footnote{For instance, managers
may \emph{require} forecasts from expert employees under conditions
of negative expected cost.}

\section{Policy and Utility Uncertainty}
\label{sec:uncertain}

We now relax two key assumptions underlying our compensation
rule from Section~\ref{sec:model}: that $E$ knows DM's policy, and
that DM knows $E$'s utility function.

\subsection{Policy Uncertainty}
\label{sec:policyuncertain}


We first consider the case where DM does not want to disclose its
policy to $E$. For example, suppose DM wanted to forego a truthful
compensation rule $C$ and simply rely on a proper scoring rule of 
the usual form that ignores the $E$'s inherent utility.
Thm.~\ref{thm:characterize} shows that DM cannot prevent misreporting
in general if it ignores $E$'s inherent utility; hence it
can suffer a loss in its own utility. However, by 
refusing to disclose its policy $\pi$, DM could reduce the 
incentive for 
$E$ to misreport. Without accurate knowledge of $\pi$, $E$ would
be forced to rely on uncertain beliefs about $\pi$ to determine
the utility of a misreport, generally lowering its 
incentive. However, this will not remove the 
misreporting incentive completely.
For instance, referring to Fig.~\ref{fig:expertUtility}, suppose
DM does not disclose $\pi$. If $E$ believes 
\emph{with sufficient probability} that the
decision boundary between $d_3$ and $d_1$ is located at the
point indicated, it will
misreport any forecast $\bfp$ in region
$D_3$ sufficiently close to that boundary should DM use a 
scoring rule rather than a compensation rule. As such,
refusing to disclose its policy can be used to reduce, but
not eliminate, the incentive to misreport if DM does not want
to use a proper compensation rule.\footnote{A similar argument
shows that a stochastic policy can be used to reduce misreporting
incentive, e.g., the \emph{soft max} policy that sees DM take
decision $d_i$ with probability proportional to $e^{\lambda u_i(\bfp)}$.
\label{fn:softmax}.}

Our analysis in the previous section assumed
that $E$ used it knowledge of $\pi$ to determine the report
that maximizes her net score. However, DM does not need to
disclose $\pi$ to make good use of a compensation rule.
It can specify a 
compensation rule \emph{implicitly} by announcing its
net scoring function $S(\bfp,x_i)$ (or the
cost function $G$ and subgradient $G^*$) and
promising to deduct $B_d\cdot\bfp$ from this score for 
whatever decision $d$ it ultimately takes. $E$ \emph{need not know}
in advance what decision will be taken to be incentivized to
offer a truthful forecast. Nor does $E$ ever need to know what
decisions \emph{would have been taken} had it reported differently.
Thus the only information $E$ needs to learn about $\pi$ is the
value of $\pi(\bfp)$ at its reported forecast $\bfp$; and even
this need not be revealed until after the decision is
taken (and its outcome realized).\footnote{Some mechanism to 
\emph{verify} the
decision \emph{post hoc} may be needed in some circumstances, but this
is no different than requiring verification of the realized
outcome in standard models of scoring rules.}

\subsection{Uncertainty in Expert Utility}
\label{sec:utilityuncertain}

We now consider the more interesting issues that arise
when DM is uncertain about the parameters $\bfb$ of $E$'s
utility function. If the DM has a distribution over $\bfb$, one 
obvious technique is to specify a proper compensation rule using the 
expectation of $\bfb$. This may work reasonably well in practice,
depending on the nature of the distribution; but it follows
immediately from Thm.~\ref{thm:characterize} that this approach will
not induce truthful reporting in general. 

Rather than analyzing probabilistic beliefs, we instead suppose that
DM has constraints on $\bfb$ that define a bounded feasible region
$\calB \subseteq \bbR^{mn}$ in which $E$'s utility parameters
must lie. We will confine our analysis to a simple, but natural
class of constraints, specifically, upper and lower bounds
on each utility parameter; i.e., assume DM has upper and 
lower bounds $\uppr{b_{ij}}$ and $\lowr{b_{ij}}$, respectively,
on each $b_{ij}$. This induces a hyper-rectangular feasible region
$\calB$. If $\calB$ is a more general region (e.g., a polytope
defined by more general linear constraints), our analysis below
can be applied to the tightest ``bounding box'' of 
the feasible region.\footnote{General linear
constraints on $E$'s parameters could be
could be inferred, for example,  from observed behavior.}
Again by Thm.~\ref{thm:characterize} it is clear that
DM cannot define a proper compensation rule in general: without 
certain knowledge of $E$'s utility, any proposed ``deduction''
of inherent utility from $E$'s compensation could mistaken,
leading to an incentive
to misreport.  However, this incentive can be bounded.

Under conditions of utility uncertainty, it is
natural for DM to restrict its attention to ``consistent''
compensation rules:
\begin{defn} Let $\calB$ be the set of feasible expert
utility functions.
A compensation rule is \emph{consistent} with $\calB$ iff
it has the form, for some (strictly) convex $G$:
\begin{align}
C(\bfp,x_i) = G(\bfp) - G^*(\bfp)\cdot\bfp + G^*_i(\bfp) - \tilde{b}_{i,\pi(\bfp)}
\label{eq:consistentcompensation}
\end{align}
for some $\tilde{\bfb}\in\calB$.
\end{defn}
Notice that consistent compensation rules are naturally linear:
intuitively, we select a single consistent estimate of
each parameter $\tilde{b_{ij}}\in [\lowr{b_{ij}}, \uppr{b_{ij}}]$,
treat $E$ as if this were her true (linear) utility function,
and define $C$ using this estimate. Let's
say DM is \emph{$\delta$-certain} of $E$'s utility iff
$\uppr{b_{ij}} - \lowr{b_{ij}} \leq \delta$ for all $i,j$.
Then we can bound the incentive for $E$ to misreport as follows:
\begin{theorem}
\label{thm:consistentbound}
If DM is $\delta$-certain of $E$'s utility, then
$E$'s incentive to misreport under any consistent compensation rule
is bounded by $2\delta$. That is, $S(\bfr,\bfp) - S(\bfp,\bfp) \leq 2\delta$.
\end{theorem}
\begin{proof}
Let $\bfp$ be $E$'s actual beliefs and $\bfr$ some report.
\begin{align*}
S(\bfr,\bfp) 
  &= [G(\bfr) - \tilde{\bfb}_{\pi(\bfr)} + \bfb_{\pi(\bfr)}]\cdot\bfp\\
 &\leq G(\bfr)\cdot\bfp + \delta\\
 &\leq G(\bfp)\cdot\bfp + \delta\\
 &\leq [G(\bfp) - \tilde{\bfb}_{\pi(\bfp)} 
        + \bfb_{\pi(\bfp)} + \delta ]\cdot\bfp + \delta\\
 &\leq S(\bfp,\bfp) + 2\delta
\end{align*}
\end{proof}
Notice that the proof assumes that: (a) the estimated utility
$\tilde{\bfb}_{\pi(\bfr)}$ for the decision induced by $E$'s
report $\bfr$ underestimates her true utility by $\delta$; and
(b) the estimated utility $\tilde{\bfb}_{\pi(\bfp)}$ for the 
optimal decision overestimates $E$'s true utility by $\delta$.
We can limit the misreporting incentive further by using a
\emph{uniform} compensation rule.
\begin{defn} 
A consistent compensation rule is \emph{uniform} 
if each parameter is estimated by
$\tilde{b}_{i,\pi(\bfp)} = \lambda\lowr{b_{ij}} + (1-\lambda)\uppr{b_{ij}}$ for
some fixed $\lambda\in [0,1]$.
\end{defn}
For example, if DM uses the lower bound (or midpoint, or upper bound, etc.) of 
each parameter interval uniformly, we call its compensation rule uniform.
\begin{corollary}
\label{cor:uniformbound}
If DM is $\delta$-certain of $E$'s utility, then
$E$'s incentive to misreport under any uniform compensation rule
is bounded by $\delta$. That is, $S(\bfr,\bfp) - S(\bfp,\bfp) \leq \delta$.
\end{corollary}

While bounding the \emph{incentive} to misreport is somewhat useful, 
it is more important to understand the impact such misreporting
can have on DM. Fortunately, this
too can be bounded. The (strict) convexity of $G$ means that
the greatest incentive to misreport occurs at the decision
boundaries of DM's policy $\pi$ in Thm.~\ref{thm:consistentbound}.
Since, by definition, DM is \emph{indifferent} between the 
adjacent decisions at any decision boundary, misreports in
a bounded region around decision boundaries have limited
impact on DM's utility, as we now show. Specifically,
we show that the \emph{amount} by which $E$ will misreport
is bounded using the ``degree of convexity'' of the cost function
$G$, which in turn bounds how much loss in utility DM will realize.

\begin{defn}
Let $G$ be a convex cost function with subgradient $G^*$.  We 
say $G$ is \emph{robust relative to $G^*$ with factor $m > 0$}
iff, for all $\bfp,\bfq\in\Delta(X)$:\footnote{The definition 
of $m$-robustness can be recast using any reasonable metric, e.g.,
$L_1$-norm or KL-divergence; but the $L_2$-norm is most convenient
below when we relate robustness to strong convexity.}
\begin{align}
G(\bfq) \geq G(\bfp) + G^*(\bfp)\cdot(\bfq-\bfp) + m\enorm{\bfq-\bfp}
\label{eq:robust}
\end{align}
\end{defn}
It is not hard to see that $m$-robustness of the pair $G, G^*$
imposes a minimum ``penalty'' on any expert misreport, as a function
of its distance from her true beliefs:
\begin{observation}
\label{obs:robustloss}
Let $C$ be a proper compensation rule based on an $m$-robust cost
function $G$ and subgradient $G^*$. Let $S$ be the induced
net scoring function. Then
$$S(\bfp,\bfp) - S(\bfq,\bfp) \geq m\enorm{\bfq-\bfp}.$$
\end{observation}
Together with Thm.~\ref{thm:consistentbound}, this gives a bound on
the degree to which an expert will misreport when an uncertain
DM uses a consistent compensation rule.
\begin{corollary}
\label{cor:robustmaxlie}
Let DM be $\delta$-certain of $E$'s utility and use a
consistent compensation rule based on an $m$-robust cost function
and subgradient. 
Let $\bfp$ be $E$'s true beliefs.
Then the report $\bfq$ that maximizes $E$'s net score
satisfies $\enorm{\bfq-\bfp} \leq \frac{2\delta}{m}$.
If the compensation rule is uniform, then
$\enorm{\bfq-\bfp} \leq \frac{\delta}{m}$.
\end{corollary}
In other words, $E$'s utility-maximizing report must be within
a bounded distance of her true beliefs if DM uses an $m$-robust
cost function to define its compensation rule.

The notion of
$m$-robustness is a slight variant of the notion of \emph{strong convexity}
\cite{boyd:convexopt} in which we use the specific subgradient $G^*$ to
measure the ``degree of convexity.'' In the specific case of twice
differentiable cost function $G$, we say $G$ is 
\emph{strongly convex with factor $m$} iff $\nabla^2 G(\bfp) \succeq mI$
for all $\bfp\in\Delta(X)$; i.e., if the matrix $\nabla^2 G(\bfp) - mI$
is positive definite \cite{boyd:convexopt}.\footnote{Alternative definitions exist for
non-differentiable $G$, but we assume a twice differentiable $G$ when
discussing strong convexity and use robustness relative to a specific
subgradient $G^*$ for non-differentiable $G$.} 
$m$-convexity is a sufficient condition for
the robustness we seek.
\begin{corollary}
\label{cor:strongmaxlie}
Let DM be $\delta$-certain of $E$'s utility and use a
consistent compensation rule based on an $m$-convex, twice
differentiable cost function $G$. Let $\bfp$ be $E$'s true beliefs.
Then the report $\bfq$ that maximizes $E$'s net score
satisfies $\enorm{\bfq-\bfp} \leq \sqrt{\frac{4\delta}{m}}$.
If the compensation rule is uniform,
then $\enorm{\bfq-\bfp} \leq \sqrt{\frac{2\delta}{m}}$.
\end{corollary}
\begin{proof}
$G$'s assumed differentiability ensures its gradient 
$\nabla G$ is the unique subgradient. Since $G$ is $m$-convex, we have
$$G(\bfq) \geq G(\bfp) + \nabla G^{T}(\bfp)(\bfq-\bfp) +
  \frac{m}{2}\enormsq{\bfq-\bfp}$$
for all $\bfp,\bfq\in\Delta(X)$ (see \cite{boyd:convexopt}).
Hence $E$'s loss in compensation is at least
$\frac{m}{2}\enormsq{\bfq-\bfp}$. Since its
gain in inherent utility by misreporting is bounded
by $2\delta$ (Thm.~\ref{thm:consistentbound}), setting the former
to be no greater than $2\delta$ yields the result. Since the gain
in inherent utility under a uniform compensation rule is $\delta$,
the stronger bound follows by substituting $\delta$ for $2\delta$
in the preceding argument.
\end{proof}

Robustness---and strong convexity if we use a differentiable cost
function---allow us to globally bound the maximum degree to which
$E$ will misreport. This allows us to give a simple, global bound
on the loss in DM utility that results from its uncertainty about
the expert's utility function. Recall that DM's utility function
$U_i$ for any decision $d_i$ is linear, hence has a constant
gradient $\nabla U_i$. (We abuse notation and simply write
$\nabla U_i$ for $\nabla U_i(\bfp)$.) The function $U_i - U_j$ is
also linear, given by parameter vector $(\bfu_i - \bfu_j)$.
Let $\bfe_k$ denote the $n$-dimensional unit vector with a
1 in component $k$ and
zeros elsewhere.

\begin{theorem}
\label{thm:robustmaxloss}
Let DM be $\delta$-certain of $E$'s utility and use a
consistent compensation rule based on an $m$-robust cost function
and subgradient. Assume $E$ reports to maximize her net score.
Then DM's loss in utility relative to a truthful report by $E$
is at most 
$\max_k [\bfe_k^T \max_{i,j} \nabla(U_i - U_j)] \sqrt{n}\frac{2\delta}{m}$.
If the compensation rule is uniform, then the bound is
tightened by a factor of two.
\end{theorem}
\begin{proof}
By Cor.~\ref{cor:robustmaxlie}, $E$'s utility maximizing report 
$\bfq$ has
an $L_2$ distance at most $\frac{2\delta}{m}$ from her true beliefs
$\bfp$.
By the Cauchy-Schwartz inequality we have
$\onenorm{\bfq-\bfp} \leq \sqrt{n}\enorm{\bfq-\bfp}$, hence
bounding its max $L_1$-deviation at $\sqrt{n}\frac{2\delta}{m}$.
Then DM's loss for any (utility-maximizing) misreport is:
\begin{align*}
\bfu_{d(\bfp)}&\cdot\bfp - \bfu_{d(\bfq)}\cdot\bfp\\
  &= \bfu_{d(\bfp)}\cdot\bfq + \nabla U_{d(\bfp)}(\bfp-\bfq)  - 
     \bfu_{d(\bfq)}\cdot\bfq - \nabla U_{d(\bfq)}(\bfp-\bfq) \\
  &\leq \nabla U_{d(\bfp)}(\bfp-\bfq) - \nabla U_{d(\bfq)}(\bfp-\bfq) \\
  &\leq \nabla [U_{d(\bfp)} - U_{d(\bfq)}] (\bfp-\bfq) \\
  &\leq \max_{i,j} \nabla(U_i - U_j) (\bfp-\bfq) \\
  &\leq \max_k [\bfe_k^T \max_{i,j} \nabla(U_i - U_j)] \sqrt{n}\frac{2\delta}{m}.
\end{align*}
Here the first inequality holds by virtue of $\bfu_{d(\bfp)}\cdot\bfq
 \leq \bfu_{d(\bfq)}\cdot\bfq$ (since $\bfu_{d(\bfq)}$ is DM's optimal
decision at $\bfq$).
\end{proof}

The same proof can be adapted to strongly convex cost functions.
\begin{corollary}
\label{cor:strongmaxloss}
Let DM be $\delta$-certain of $E$'s utility and use a
linear compensation rule based on an $m$-convex, twice
differentiable cost function $G$. 
Assume $E$ reports to maximize her net score.
Then DM's loss in utility relative to a truthful report by $E$
is at most $\max_k [\bfe_k^T \max_{i,j} \nabla(U_i - U_j)]
 \sqrt{n\frac{4\delta}{m}}$.
If the compensation rule is uniform, then the bound is
tightened by a factor of two.
\end{corollary}

The results above all rely on the \emph{global} robustness or 
\emph{global} strong convexity
of the cost function $G$. Designing a specific cost function (and if
not differentiable, choosing its subgradients) can be challenging if
we try to ensure uniform $m$-robustness or $m$-convexity
across the entire probability space $\Delta(X)$.
But recall that $E$ can only impact DM's
utility if its misreport causes DM to change its decision. This means that
the cost function need only induce strong penalties for misreporting
near decision boundaries. Furthermore, the strength of these penalties
should be related to the rate at which DM's utility is negatively impacted.
For example, suppose $\bfp$ lies on the decision boundary between
region $D_i$ and $D_j$. If $|\nabla(U_i - U_j)|$ is large,
then a misreport in the region around $\bfp$ will cause a greater
loss in utility than if $|\nabla(U_i - U_j)|$ is small. This suggests
that the cost function should be more strongly convex (or more robust)
near decision boundaries whose corresponding decisions differ significantly
in utility, and can be less strong when the decisions are ``similar.''
See Fig.~\ref{fig:localconvexity} for an illustration of this point.
\begin{figure}
\centering
\includegraphics[scale=0.44]{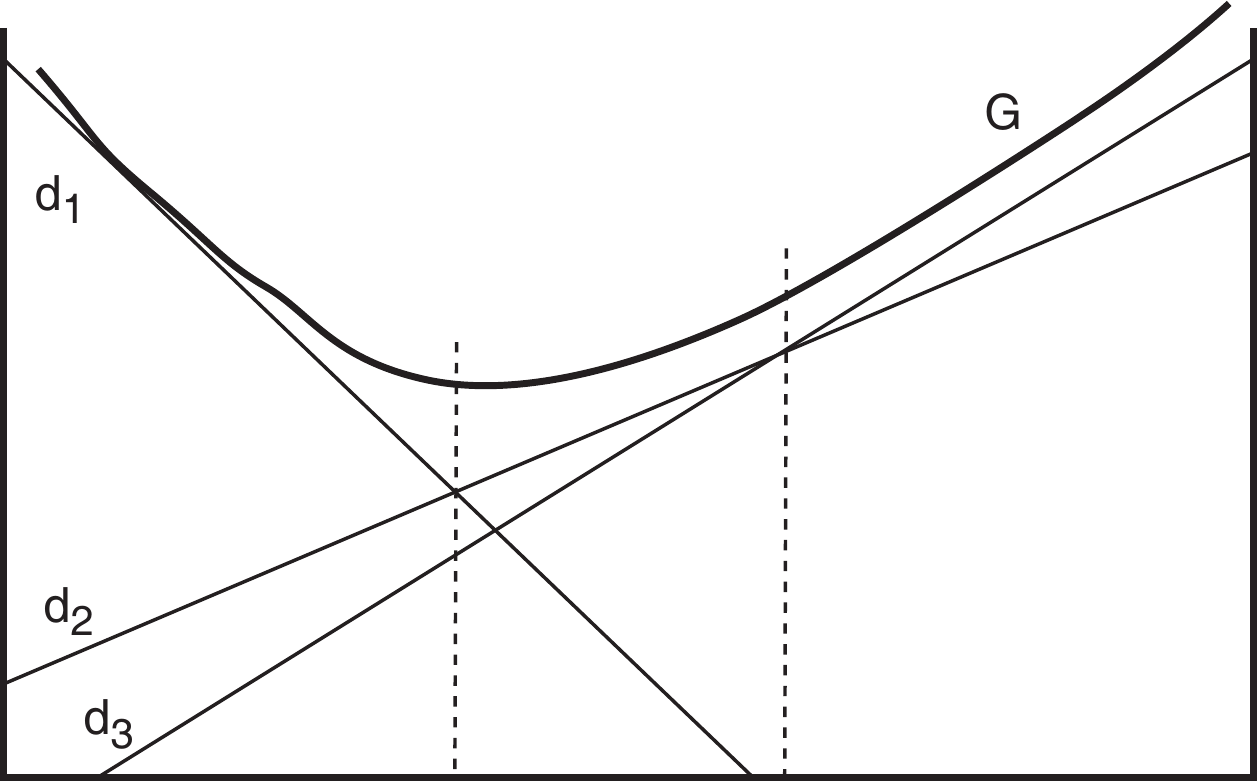}
\vskip -2mm
\textrm\small{\caption{A locally strongly convex cost function $G$.
Here $G$ has is more strongly convex in the neighborhood
of decision boundary $D_{12}$ than the boundary $D_{23}$.
This means an expert willing to sacrifice compensation (e.g.,
to gain inherent utility due to DM uncertainty) can offer a report 
that deviates more from its true beliefs in the neighborhood around
$D_{23}$ and than it can in the neighborhood around $D_{12}$ for
the same loss in compensation. However, since $d_2$ and $d_3$ are
more similar than $d_1$ and $d_2$ (w.r.t.\ DM utility), i.e.,
the gradient $|\nabla(U_2 - U_3)|$ is less than 
$|\nabla(U_1 - U_2)|$, DM will lose less utility ``per unit''
of misreport in the neighborhood of $D_{23}$. Note:
the cost function $G$ is drawn above DM's utility function
$U$ for illustration only---in general, it will lie below $U$.
}
\label{fig:localconvexity}
}
\vskip -2mm
\end{figure}
Furthermore, the cost function need only be robust or strongly convex
in a \emph{local region} around these decision boundaries. In particular,
suppose $G$ is $m$-robust in some local region around the decision
boundary between $D_i$ and $D_j$. The degree of robustness bounds
the maximum deviation from truth that $E$ will contemplate. If
the region of $m$-robustness includes these maximal deviations,
that will be sufficient to bound DM's utility loss for any true beliefs
$E$ has in that region. Outside of these regions, no misreport by $E$
will cause DM to change its decisions (relative to a truthful
report).

We can summarize this as follows:
\begin{defn}
$G$ is \emph{locally robust relative to $G^*$} 
in the $\veps$-neighborhood around $\bfp$ with factor $m > 0$
iff, for all $\bfq\in\Delta(X)$ s.t.\ $\enorm{\bfq-\bfp} \leq \veps$:
\begin{align}
G(\bfq) \geq G(\bfp) + G^*(\bfp)\cdot(\bfq-\bfp) + m\enorm{\bfq-\bfp}
\label{eq:localrobust}
\end{align}
Local strong convexity is defined similarly.
\end{defn}

Now suppose DM wishes to bound its loss due to misreporting by $E$
by some factor $\sigma > 0$. This can be accomplished using a
locally robust cost function:
\begin{theorem}
\label{thm:localrobustmaxloss}
Let DM be $\delta$-certain of $E$'s utility and fix $\sigma > 0$.
For any pair of decisions $d_i, d_j$ with non-empty decision boundary
$D_{ij}$, define
{\small
$$m_{ij} =\! \frac{\max_k (\bfe_k^T \nabla[U_i\! -\! U_j]) \sqrt{n}2\delta}{\sigma};\;\;
  \veps_{ij} =\! \frac{\sigma}{\max_k (\bfe_k^T \nabla[U_i\! -\! U_j]) \sqrt{n}}.$$
}
Let $G$ be a convex cost function with subgradient $G^*$
such that, for all $i,j$ and any $\bfp \in D_{ij}$, 
  (a) $G$ is locally robust with factor $m_{ij}$ in the 
  $\veps_{ij}$-neighborhood around $\bfp$; (b) no other decision boundary
  lies within the $\veps_{ij}$-neighborhood around $\bfp$.
Let DM use a consistent compensation rule based on $G, G^*$.
Assume $E$ reports to maximize her net score.
Then DM's loss in utility relative to a truthful report by $E$
is at most $\sigma$. If the compensation rule is uniform, the result
holds with both $m_{ij}$ and $\veps_{ij}$ decreased by a factor of two.
\end{theorem}
\begin{proof}
(Sketch). The proof proceeds by cases involving the location of
$E$'s true beliefs $\bfp$ and the location of possible utility-maximizing
misreports $\bfr$. W.l.o.g., assume that $\bfp$ is in decision
region $D_i$. We consider four classes of misreports.

(A) Suppose $\bfr \in D_i$. In this case, DM's utility loss is
zero since the decision is the same as if $E$ had reported truthfully.

(B) Now consider the case where $\bfp$ is in
the $\veps_{ij}$-neighborhood of some decision boundary $D_{ij}$.
We show that any report $\bfr \in D_j$ satisfies the condition
of the theorem. Let $\bfq \in D_j$ be
an arbitrary point s.t.\ $\enorm{\bfq-\bfp} \leq \veps_{ij}$ 
(this must exist by the assumption that $\bfp$ lies in the
$\veps_{ij}$-neighborhood of $D_{ij}$). 
DM's utility loss for reporting $\bfq$ is then bounded as follows:
\begin{align*}
\bfu_i\cdot\bfp &- \bfu_j\cdot\bfp\\
  &= \bfu_i\cdot\bfq + \nabla U_i(\bfp-\bfq)  - 
     \bfu_j\cdot\bfq + \nabla U_j(\bfp-\bfq) \\
  &\leq \nabla [U_i - U_j] (\bfp-\bfq) \\
  &\leq \max_k [\bfe_k^T \nabla(U_i - U_j)]\onenorm{\bfp-\bfq}\\
  &\leq \max_k [\bfe_k^T \nabla(U_i - U_j)]
            \sqrt{n}\enorm{\bfp-\bfq}\\
  &\leq \max_k [\bfe_k^T \nabla(U_i - U_j)]
            \sqrt{n}\veps_{ij}\\
  &= \sigma.
\end{align*}
If $\bfr\in D_j$, then it must be such a $\bfq$ (i.e.,
be within $\veps_{ij}$ of $\bfp$), since any report in $D_j$ has the same
inherent utility, while those closest to $\bfp$ maximize compensation.
Hence utility loss for $\bfr$ is no greater than $\sigma$.
Note that $\bfp$ may lie within the $\veps_{ij}$ neighborhood of multiple
decision boundaries $D_{ij}$ adjacent to $D_i$, but the argument holds for
any report in any such region $D_j$.

(C) Now consider the case where boundary $D_{ij}$ exists,
but $\bfp$ does not lie within the
$\veps_{ij}$-neighborhood of $D_{ij}$. We show that 
$E$'s utility maximizing report cannot be in $D_j$. By way of contradiction,
consider a report $\bfr\in D_j$. Let $\ell$ be the closed line segment
$\{(1-\lambda)\bfp + \lambda\bfr : \lambda\in [0,1]\}$; and 
let $\bfq\in D_{ij}$ be
the point where $\ell$ intersects the decision boundary, and let
$\bfp'\in D_i$ be the point on $\ell$ on the ``$D_i$ side'' of
the boundary that is distance $\veps_{ij}$ from the boundary.
$E$'s loss in net score (ignoring any error due inherent utility
misestimate by DM) is given by:
\begin{align*}
S(\bfp,&\bfp) - S(\bfr,\bfp)\\
  &= H_\bfp\cdot\bfp - H_\bfr\cdot\bfp\\
  &= (H_\bfp\!\cdot\!\bfp - H_{\bfp'}\!\cdot\!\bfp)
     + (H_{\bfp'}\!\cdot\!\bfp - H_\bfq\!\cdot\!\bfp)
     + (H_\bfq\!\cdot\!\bfp - H_\bfr\!\cdot\!\bfp)
\end{align*}
We have $(H_\bfp\cdot\bfp - H_{\bfp'}\cdot\bfp) \geq 0$ by
the propriety of the compensation rule (ignoring error due to
inherent utility misestimation). We also have
\begin{align*}
H_{\bfp'}\cdot\bfp &- H_\bfq\cdot\bfp\\
 &= H_{\bfp'}\cdot(\bfp -\bfp' +\bfp') - H_\bfq\cdot(\bfp -\bfp' +\bfp')\\
 &= H_{\bfp'}\cdot\bfp' - H_\bfq\cdot\bfp'
     + H_{\bfp'}\cdot(\bfp -\bfp') - H_\bfq\cdot(\bfp -\bfp')\\
 &\geq m_{ij}\veps_{ij}
     + H_{\bfp'}\cdot(\bfp -\bfp') - H_\bfq\cdot(\bfp -\bfp')\\
 &\geq m_{ij}\veps_{ij}\\
 &\geq 2\delta
\end{align*}
where the first inequality holds due to the local robustness of
$G$ and the second due to the convexity of $G$ and the
collinearity of $(\bfp,\bfp',\bfq)$.
Finally, we must have $(H_\bfq\cdot\bfp - H_\bfr\cdot\bfp) \geq 0$
again due to the convexity of $G$ and the collinearity of $(\bfp,\bfq,\bfr)$.
Thus $E$'s loss in compensation due to misreporting is at least
$2\delta$ (and is strictly greater if $G$ is strictly convex). But 
by Thm.~\ref{thm:consistentbound} its gain in inherent utility by
misreporting can be no greater than $2\delta$. Hence its optimal
report $\bfr$ cannot lie in $D_j$.

(D) The preceding argument can be adapted in a straightforward way to
the case where $D_i$ and $D_j$ are not adjacent (i.e., $D_{ij}$ is
empty).
\end{proof}
This result can be generalized to the case where the degree of robustness
around one decision boundary is relaxed sufficiently so that the
neighborhood within which $E$ can profitably misreport crosses
more than one decision boundary (i.e., when 
another decision boundary overlaps the $\veps_{ij}$-neighborhood
around $D_{ij}$).  Utility loss 
will increase but is can be bounded by considering the maximum 
gradient $\nabla(U_i - U_j)$ over decisions that can be swapped.
The result can also be adapted to locally strongly convex cost functions in
the obvious way.
\begin{corollary}
\label{cor:localstrongmaxloss}
Let DM be $\delta$-certain of $E$'s utility and fix $\sigma > 0$.
For any pair of decisions $d_i, d_j$ with non-empty decision boundary
$D_{ij}$, define
\begin{small}
$$m_{ij} =\! \frac{\max_k (\bfe_k^T \nabla[U_i \!- \! U_j]) \sqrt{n}2\delta}{\sigma};\;\;
  \veps_{ij} =\! \frac{\sigma}{\max_k (\bfe_k^T \nabla[U_i\! -\! U_j]) \sqrt{n}}.$$
\end{small}
Let $G$ be a convex cost function
such that, for all $i,j$ and any $\bfp \in D_{ij}$, 
  (a) $G$ is locally convex with factor $m_{ij}$ in the 
  $\veps_{ij}$-neighborhood around $\bfp$; (b) no other decision boundary
  lies within the $\veps_{ij}$-neighborhood around $\bfp$.
Let DM use a consistent compensation rule based on $G, G^*$.
Assume $E$ reports to maximize her net score.
Then DM's loss in utility relative to a truthful report by $E$
is at most $\sigma$.
\end{corollary}
These results quantify the ``cost'' to the decision maker of its
imprecise knowledge of the expert's utility function, i.e.,
its worst-case expected utility relative to what it could have
achieved if it had full knowledge of $E$'s utility (i.e.,
with truthful reporting by $E$).

\begin{remark} 
If we relax the constraint that DM choose the decision
$d_i$ with maximum expected utility, we can exploit local
robustness to induce truthful forecasts. Suppose DM uses
the softmax decision policy (see footnote~\ref{fn:softmax}): this stochastic
policy makes $E$'s utility $B^\pi(\bfr,\bfp)$ continuous in its
report $\bfr$. An analysis similar to that above, using local 
robustness or strong convexity of the cost function, allows DM
to induce truthtelling as long as the degree of convexity
compensates for the gradient of $B^\pi$ at decision boundaries.
Since adding randomness to the policy removes the 
discontinuities in $B^\pi$, this is
now possible. 
Of course, this ``incentive-compatibility'' comes at a cost:
the DM is committed to taking suboptimal actions with
some probability. We defer a full analysis of the tradeoffs,
and the relative benefits of ``acting optimally'' but risking
misleading reports vs.\ ``acting suboptimally'' relative to
truthful report, to a longer version of this paper.
\end{remark}

The characterization of DM loss using local robustness or local
strong convexity not only offers theoretical guarantees on
DM utility---it has potential operation significance in the design
of compensation rules. Specifically, it suggests an optimization procedure 
for designing a cost function $G$---from which the induced compensation
rule $C$ is recovered---so as to minimize DM utility loss. Intuitively,
the design of $G$ will attempt to optimize two conflicting objectives:
minimizing the bound $\sigma$ on utility loss, which generally requires
\emph{increasing} the degree of robustness or convexity of $G$ at decision
boundaries; and minimizing expected compensation $c$ which, given the
requirement of strict convexity of $G$, generally requires \emph{decreasing}
robustness or convexity. This tension can be addressed by either:
(a) explicitly trading $\sigma$ and $c$ off against each other in the 
design objective; (b)
minimizing $c$ subject to a target bound $\sigma$; or (c) minimizing
$\sigma$ subject to a target compensation level $c$. The optimization
itself is defined over the space of $n$-dimensional convex curves $G$,
and could be treated as an $n$-dimensional spline problem.
The objective is to fit a convex function to
a set of points with specific local curvature constraints that
enforce a certain degree of local convexity at particular
decision boundaries. Specific classes of
spline functions (e.g., Catmull-Rom splines) might prove useful for
this purpose.  We leave to future research the question of the
practical design of cost and compensation functions under conditions of utility
uncertainty.

\section{Market Scoring Rules}
\label{sec:msr}

Space precludes a comprehensive treatment, but we provide a brief
sketch of how one might exploit compensation functions
in settings where DM aggregates the forecasts of multiple experts.
One natural means
of doing so is to develop a \emph{market scoring rule (MSR)}
\cite{hanson-scoringrules:isf03,hanson-scoringrules:jpm07}
that sequentially applies a standard scoring rule based on
how an expert alters the prior forecast (see Sec.~\ref{sec:background}). 
The typical means of creating an MSR given a scoring
rule $S$ is to have the $k$th expert (implicitly) pay the 
$k{\! -\!}1$st expert for its forecast according to $S$, and have the 
principal pay only final expert for its forecast using $S$. In this
way, the principal's total payment is bounded by the maximal
possible payment to a single expert \cite{hanson-scoringrules:jpm07}.

When one attempts this with self-interested experts,
difficulties emerge. For instance, Shi et al.~\cite{Shi-Conitzer:wine09}
show that experts who can alter the outcome 
distribution after making a forecast, \emph{each} require compensation 
to prevent them from manipulating the distribution in ways that
are detrimental to the principal.\footnote{Shi 
et al.~\cite{Shi-Conitzer:wine09} actually use a one-round variant of 
an MSR.} A related form of \emph{subsidy} arises in our decision setting.

Following~\cite{Shi-Conitzer:wine09}, 
we assume a collection
of $n$ experts, each of whom can provide alter the forecast $\bfp$
exactly once. 
Suppose the experts have an interest in DM's
decision. An ``obvious'' MSR in our model would
simply adopt a proper compensation rule, and have each expert pay
the either the \emph{compensation} or the \emph{net score}
due to the expert who provided the incumbent forecast,
and receive her payment from the next expert. If we use
compensation, we run into strategic issues. With
a proper compensation rule, an expert $k$
reports truthfully based on her
\emph{net score (total utility)}, consisting of both compensation
and the inherent utility of the decision she induces. In a
market setting, $k$'s proposed decision may be \emph{changed} by the next
expert that provides a forecast. This
(depending on her beliefs about other expert opinions) may
incentivize $k$ to misreport in order to \emph{maximize her
compensation} rather than her net score.  
Overcoming such strategic issues seems challenging.

Alternatively, each expert might pay the net score due her
predecessor. Unfortunately, an arbitrary proper
compensation rule may not pay expert $k$ enough
score to ``cover her costs'' (e.g., if $k{\! -\!}1$'s inherent utility is 
much higher than $k$'s). However, if we set aside issues
associated with incentive for participation for the moment,
the usual MSR approach can be adapted as follows:
we fix a \emph{single
(strictly) convex cost function $G$ for all experts}, and define 
the compensation rule $C^k$ for expert $k$ using $G$ in
the usual way:
$$C^k(\bfp,x_i) = G(\bfp) - G^*(\bfp)\cdot\bfp 
   + G^*_i(\bfp) - b^k_{i,\pi(\bfp)},$$
where $\bfb^k$ is $k$'s utility function (bias). If
$G$ satisfies strong participation for all experts
(i.e., if $G(x_i) \geq B^*(x_i)$ for all $i$), then any
expert $k$ whose beliefs $\bfp[k]$ differ from the forecast $\bfp[k{\! -\!}1]$
provided by $k{\! -\!}1$ will have an expected net score (given
$\bfp[k]$) greater than her expected payment to $k{\! -\!}1$ and
will maximize her utility by providing a truthful forecast.
In particular, let's denote $k$'s expected payment to $k{\! -\!}1$
by $\rho(k, k{\! -\!}1)$; then we have:
\begin{align*}
\rho(k, k{\! -\!}1) &= 
 (H_{\bfp[k{\! -\!}1]} - \bfb^{k{\! -\!}1}_{\pi(\bfp[k{\! -\!}1])})\cdot\bfp[k] 
 \, +  \, \bfb^{k{\! -\!}1}_{\pi(\bfp[k{\! -\!}1])}\cdot\bfp[k] \\
 &= H_{\bfp[k{\! -\!}1]} \cdot\bfp[k] \\
 &\leq H_{\bfp[k]} \cdot\bfp[k].
\end{align*}
Hence $k$'s expected payment $\rho(k, k{\! -\!}1)$ is
less than its expected net utility, leaving it with a
(positive) net gain of $(H_{\bfp[k]} - H_{\bfp[k{\! -\!}1]}) \cdot\bfp[k]$.
However, this gain may be smaller than the inherent
utility she derives from the decision induced by $k{\! -\!}1$,
namely, $\bfb^{k}_{\pi(\bfp[k{\! -\!}1])}\cdot\bfp[k]$.
Hence this scheme may not incentivize participation. In
cases where DM can force participation, such a scheme
can be used; but in general, the self-subsidizing nature of
standard MSRs cannot be exploited with self-interested 
experts.\footnote{If expert utility is small
relative to overall compensation, we can exploit the
strong robustness (or strong convexity) of the cost function
to show that experts will abstain from offering
predictions only if their beliefs are sufficiently close
to the incumbent prediction. Providing the 
degree of compensation induced by an ``extremely convex''
cost function can, of course, be interpreted as a form of
subsidy.}

To incentivize participation, 
DM can subsidize these payments. In the most extreme case, 
DM simply pays each displaced expert its net utility, which removes any
incentives to misreport, but at potentially high cost. In
certain circumstances, we can reduce the DM subsidy to the
market by having it pay 
only the inherent utility $b^{k{\! -\!}1}_{i,\pi(\bfp[k{\! -\!}1])}$
(given realized outcome $x_i$)
of the displaced expert $k{\! -\!}1$, and 
requiring the displacing expert $k$ to pay the compensation
$H_{i,\bfp[k{\! -\!}1]}$. Under certain conditions on
the relative utility of different experts for different decisions, this 
is sufficient to induce participation; that is, $k$'s net gain for
partipating exceeds her inherent utility for the incumbent decision.

For instance, suppose all experts
have the same utility function $\bfb$ (e.g., consider experts
in the same division of a company who are asked to predict the
outcome of some event, and have different estimates, but have
aligned interests in other respects). In this case, $k$'s net gain for
reporting her true beliefs is:
\begin{align*}
(H_{\bfp[k]} 
   &- ( H_{\bfp[k{\! -\!}1]} - \bfb_{\pi(\bfp[k{\! -\!}1])} ))\cdot\bfp[k] \\
  &= (H_{\bfp[k]} - H_{\bfp[k{\! -\!}1]})\cdot\bfp[k] 
     + \bfb_{\pi(\bfp[k{\! -\!}1])}\cdot\bfp[k] \\
  &\geq \bfb_{\pi(\bfp[k{\! -\!}1])}\cdot\bfp[k].
\end{align*}
Hence $k$'s expected net gain is at least as great as her inherent expected
utility for the decision induced by $k{\! -\!}1$, and strictly greater
if her beliefs differ from those of $k{\! -\!}1$. Thus
participation is assured.

Indeed, the argument holds even if the utility functions are not
identical: we require only that $k$'s expected utility for the
decision it displaces is less than the expected utility (given $k$'s
beliefs) to be offered to her predecessor $k{\! -\!}1$. A sufficient condition
for this is that $\bfb^k \leq \bfb^{k{\! -\!}1}$ (pointwise). This
suggests that if the DM can elicit predictions of its experts
in a particular order, it should do so by eliciting forecasts of those 
with the greatest utility first.

In general, even in the extreme case of identical expert utility
functions, there seems to be no escape from the requirement that
DM subsidize the market, at a level that grows linearly with
the number of agents. This is very similar to the conclusions
drawn by Shi et al.~\cite{Shi-Conitzer:wine09}. We provide a more
detailed formalization and analysis in an extended version
of the paper.

\section{Concluding Remarks}
\label{sec:conclude}

We have presented a model that allows the analysis of the incentives
facing experts in a decision-making context who have a vested interest in
the decision taken by the principal. We have developed a class of
\emph{compensation rules} that are necessary and sufficient to induce
truthful forecasts from self-interested experts and also characterized
the subclasses of such rules that satisfy weak and strong participation
constraints. While vanilla compensation rules assume knowledge of
the expert's utility function on the part of the principal, we've
also shown how to design compensation rules when the principal has
only rough bounds on expert utility parameters in such a way that
(a) the incentive for the expert to misreport is bounded; and
(b) the impact on the principal's decision/utility is similarly
bounded relative to the case of full knowledge. These bounds are
derived from the robustness or strong convexity of the cost function, 
in either a global or a local sense.

A number of other interesting directions remain. One is the
development of computationally effective procedures to design cost 
functions that minimize principal utility loss in settings where
expert utility is not fully known, without inducing extreme degrees
of compensation. 
Another direction is the analysis of the
tradeoff between the principal's decision space and the payments
required to induce truthfulness or participation on the part of
experts. We discussed above the possibility that by acting
(somewhat) suboptimally through the use of a stochastic policy,
the principal could diminish the incentive for the expert to
misreport. We can take this a step further: intuitively, by limiting 
its policy to use only a \emph{subset} of its potential decisions, 
the principal may dramatically reduce
the potential for strategic behavior---either misreporting or failure
to participate---on the part of an expert, while
at the same time, doing little damage its own utility by restricting
its policy in this way.\footnote{Thanks to Tuomas Sandholm for suggesting
this and preliminary discussion in this direction.} 

Finally, the question of joint elicitation both the utility function and 
the forecast of an expert remains intriguing. This can be viewed as
a mechanism design problem where the expert's type consists of
both its preferences and its ``information'' or forecast (see
\cite{Conitzer:uai09} for a treatment of forecasts themselves from a
mechanism design perspective). Preliminary
results (joint with Tuomas Sandholm) suggest, not surprisingly,
that truthful elicitation is not generally possible when the
principal takes a decision that maximizes its expected utility. However, 
it remains to be seen if effective mechanisms can be designed that
offer ``reasonable'' performance from the perspective of the
principal.

\subsection*{Acknowledgements} 
Thanks to Yiling Chen, Vince Conitzer, Ian Kash, 
Tuomas Sandholm and Jenn Wortman Vaughan
for helpful discussions and suggestions. This work 
was supported by NSERC.

\bibliographystyle{plain}

\end{document}